\documentclass[12pt]{article}
\usepackage{amsmath,amssymb}
\usepackage{amsmath,amsthm}
\usepackage{graphicx}
\usepackage{enumitem}
\usepackage{authblk}
\usepackage{color}
\usepackage[authoryear]{natbib}
\DeclareMathOperator*{\argmin}{arg\,min}

\usepackage[utf8]{inputenc} 
\usepackage[T1]{fontenc}    
\usepackage{lmodern}
\usepackage{hyperref}       
\usepackage{url}            
\usepackage{booktabs}       
\usepackage{amsfonts,amsmath,amsthm,amssymb}       
\usepackage{nicefrac}       
\usepackage{microtype}      

\newtheorem{Th}{Theorem} 

\theoremstyle{definition}

\theoremstyle{remark}
\newtheorem{rmk}{Remark}

\usepackage{geometry}
\title{Blind nonnegative source separation using biological neural networks}

\author[1]{Cengiz Pehlevan} 
\author[1,2]{Sreyas Mohan} 
\author[1,3]{Dmitri B. Chklovskii} 
\affil[1]{Center for Computational Biology, Flatiron Institute, New York, NY}
\affil[2]{IIT Madras, Chennai, India}
\affil[3]{NYU Medical School, New York, NY}
 
 \date{\vspace{-5ex}}
 
\begin{document}

\maketitle

\begin{abstract}
Blind source separation, i.e. extraction of independent sources from a mixture, is an important problem for both artificial and natural signal processing. Here, we address a special case of this problem when sources (but not the mixing matrix) are known to be nonnegative, for example, due to the physical nature of the sources. We search for the solution to this problem that can be implemented using biologically plausible neural networks. Specifically, we consider the online setting where the dataset is streamed to a neural network. The novelty of our approach is that we formulate blind nonnegative source separation as a similarity matching problem and derive neural networks from the similarity matching objective. Importantly, synaptic weights in our networks are updated according to biologically plausible local learning rules. 
\end{abstract}

\section{Introduction}

Extraction of latent causes, or sources, from complex stimuli is essential for making sense of the world. Such stimuli could be mixtures of sounds, mixtures of odors, or natural images.  If supervision, or ground truth, about the causes is lacking, the problem is known as blind source separation. 

The blind source separation problem can be solved by assuming a generative model, wherein the observed stimuli are linear combinations of independent sources, an approach known as Independent Component Analysis (ICA) \citep{jutten1991blind,comon1994independent,bell1995information,hyvarinen2000independent}. Formally, the stimulus at time $t$ is expressed as a $k$-component vector
\begin{align}\label{genmod}
{\bf x}_t ={\bf A} {\bf s}_t,
\end{align}
where ${\bf A}$ is an unknown but time-independent $k\times d$ mixing matrix and ${\bf s}_t$ represents the signals of $d$ sources at time $t$. In this paper we assume that $k\geq d$. 

The goal of ICA is to infer source signals, ${\bf s}_t$, from the stimuli ${\bf x}_t$. Whereas many ICA algorithms have been developed by the signal processing community \citep{comon2010handbook}, most of them cannot be implemented by biologically plausible neural networks. Yet, our brains can solve the blind source separation problem effortlessly \citep{bronkhorst2000cocktail,asari2006sparse,narayan2007cortical,bee2008cocktail,mcdermott2009cocktail,mesgarani2012selective,golumbic2013mechanisms,isomura2015cultured}. Therefore, discovering a biologically plausible ICA algorithm is an important problem.

For an algorithm to be implementable by biological neural networks it must satisfy (at least) the following requirements. First, it must operate in the online (or streaming) setting. In other words, the input dataset is not available as a whole but is streamed one data vector at a time and the corresponding output must be computed before the next data vector arrives. Second, the output of most neurons in the brain (either a firing rate, or the synaptic vesicle release rate) is nonnegative. Third, the weights of synapses in a neural network must be updated using local learning rules, i.e. depend on the activity of only the corresponding pre- and postsynaptic neurons. 

Given the nonnegative nature of neuronal output we consider a special case of ICA where sources are assumed to be nonnegative, termed Nonnegative Independent Component Analysis (NICA),  \citep{plumbley2001adaptive,plumbley2002conditions}. Of course, to recover the sources, one can use standard ICA algorithms that don't rely on the nonnegativity of sources, such as fastICA \citep{hyvarinen1997fast, hyvarinen1999fast, hyvarinen2000independent}. Neural learning rules have been proposed for ICA, e.g. \citep{linsker1997local,eagleman2001cerebellar,isomura2016local} and references within. However, taking into account nonnegativity may lead to simpler and more efficient algorithms \citep{plumbley2001adaptive,plumbley2003algorithms,plumbley2004nonnegative,oja2004blind,yuan2004fastica,zheng2006nonnegative,ouedraogo2010non,li2016recovery}. 

While most of the existing NICA algorithms have not met the biological plausibility requirements, in terms of online setting and local learning rules, there are two notable exceptions. First, \cite{plumbley2001adaptive} succesfully simulated a neural network on a small dataset, yet no theoretical analysis was given. Second, \cite{plumbley2003algorithms} and \cite{plumbley2004nonnegative} proposed a nonnegative PCA algorithm for a streaming setting, however its neural implementation requires nonlocal learning rules. Further, this algorithm requires prewhitened data (see also below), yet no streaming whitening algorithm was given.

Here, we propose a biologically plausible NICA algorithm. The novelty of our approach is that the algorithm is derived from the similarity matching principle which postulates that neural circuits map more similar inputs to more similar outputs \citep{pehlevan2015MDS}. Previous work proposed various objective functions to find similarity matching neural representations and solved these optimization problems with biologically plausible neural networks \citep{pehlevan2015MDS,pehlevan2015normative,pehlevan2014NMF,hu2014SMF,pehlevan2015optimization}. Here we apply these networks to NICA. 

The rest of the paper is organized as follows:  In Section \ref{offline}, we show that blind source separation, after a generalized prewhitening step, can be posed as a nonnegative similarity matching (NSM) problem \citep{pehlevan2014NMF}. In Section \ref{online}, using results from \citep{pehlevan2015normative,pehlevan2014NMF} we show that both the generalized prewhitening step and the NSM step can be solved online by neural networks with local learning rules. Stacking these two networks leads to the two-layer NICA network. In Section \ref{numerical}, we compare the performance of our algorithm to other ICA and NICA algorithms for various datasets.

\section{Offline NICA via NSM}\label{offline}

In this section, we first review Plumbley's analysis of NICA  and then reformulate NICA as an NSM problem.
 
\subsection{Review of Plumbley's analysis}

When source signals are nonnegative, the source separation problem simplifies. It can be solved in two straightforward steps: noncentered prewhitening and orthonormal rotation \citep{plumbley2002conditions}. 

Noncentered prewhitening transforms ${\bf x}$ to ${\bf h} := {\bf F}{\bf x}$, where ${\bf h} \in {\mathbb R}^d$ and ${\bf F}$ is a $d \times k$ whitening matrix\footnote{In his analysis Plumbley \citep{plumbley2002conditions} assumed $k=d$ (mixture channels are the same as source channels) but this assumption can be relaxed as shown.}, such that ${\bf C}_{\bf h} := \left<\left({\bf h}-\left<{\bf h}\right>\right)\left({\bf h}-\left<{\bf h}\right>\right)^\top\right> = {\bf I}_{d}$,  where angled brackets denote an average over the source distribution and ${\bf I}_{d}$ is the $d\times d$ identity matrix. Note that the mean of ${\bf x}$ is not removed in the tranformation, otherwise one would not be able to use the constraint that the sources are nonnegative \citep{plumbley2003algorithms}. 

Assuming that sources have unit variance\footnote{Without loss of generality, a scalar factor that multiplies a source can always be absorbed into the corresponding column of the mixing matrix}, 
\begin{align}\label{white}
{\bf C}_{\bf s} :=\left<\left({\bf s}-\left<{\bf s}\right>\right)\left({\bf s}-\left<{\bf s}\right>\right)^\top \right>= {\bf I}_{d},
\end{align}
the combined effect of source mixing and prewhitening ${\bf FA}$ (${\bf h} = {\bf Fx} = {\bf FAs}$) is an orthonormal rotation. To see this, note that, by definition, ${\bf C}_{\bf h} = \left({\bf FA}\right){\bf C}_s\left( {\bf FA} \right)^{\top} = \left({\bf FA}\right)\left( {\bf FA} \right)^{\top}$ and ${\bf C}_{\bf h} = {\bf I}_{d}$. 

The second step of NICA relies on the following observation \citep{plumbley2002conditions}:
\begin{Th}[Plumbley] Suppose sources are independent, nonnegative and well-grounded, i.e. {\rm Prob}$\left(s_i<\delta\right) > 0$ for any $\delta>0$. Consider an orthonormal transformation ${\bf y} = {\bf Q}{\bf s}$. Then ${\bf Q}$ is a permutation matrix  with probability 1, if and only if ${\bf y}$ is nonnegative. 
\end{Th}
In the second step, we look for an orthonormal ${\bf Q}$ such that ${\bf y} = {\bf Q}{\bf h}$ is nonnegative. When found, Plumbley's theorem guarantees that $ {\bf Q}{\bf F}{\bf x}$ is a permutation of the sources. Several algorithms have been developed based on this observation \citep{plumbley2003algorithms,plumbley2004nonnegative,oja2004blind,yuan2004fastica}.

Note that only the sources ${\bf s}$ but not necessarily the mixing matrix ${\bf A}$ must be nonnegative. Therefore, NICA allows generative models, where features not only add up but also cancel each other, as in the presence of a shadow in an image \citep{plumbley2002conditions}. In this respect, NICA is more general than Nonnegative Matrix Factorization (NMF) \citep{lee1999learning,paatero1994positive} where both the sources and the mixing matrix are required to be nonnegative.

\subsection{NICA as NSM}

\begin{figure}
	\centering
	\includegraphics{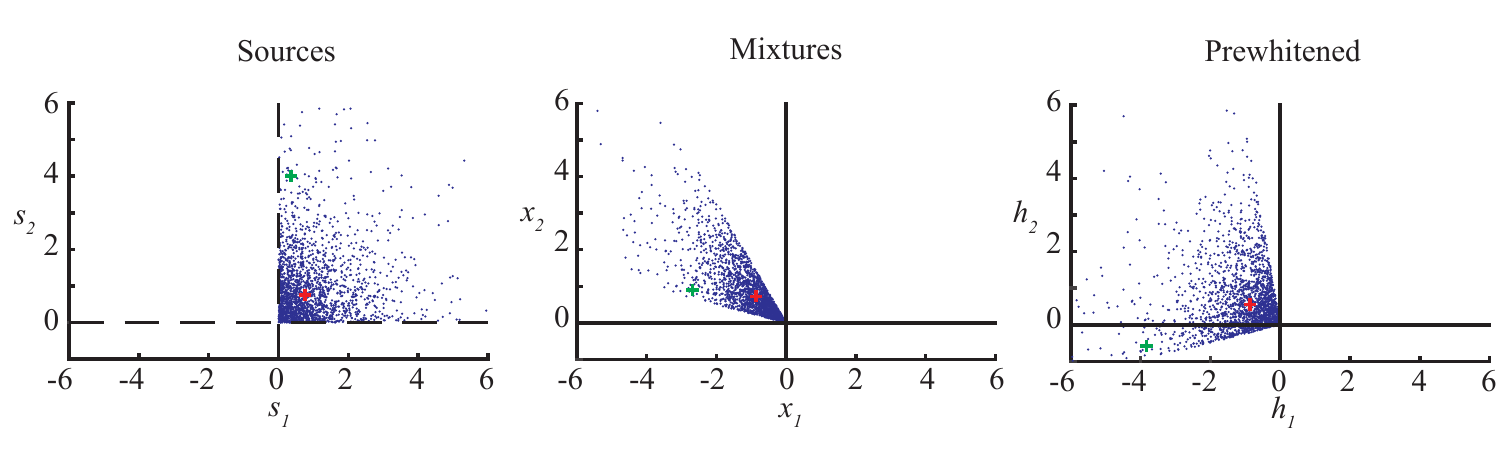}
	\caption{Illustration of the NICA algorithm. Two source channels (left) are linearly transformed to a two-dimensional mixture (middle). Prewhitening (right) gives an orthogonal transformation of the sources. Sources are then recovered by solving the NSM problem \eqref{NSM}. Green and red plus signs track two source vectors across mixing and prewhitening stages. \label{Fig1}}
\end{figure}

Next we reformulate NICA as a NSM problem. This reformulation will allow us to derive an online neural network for NICA in Section \ref{streaming}. Our main departure from Plumbley's analysis is to work with similarity matrices rather than covariance matrices and finite number of samples rather than the full probability distribution of the sources.

First, let us switch to the matrix notation, where data matrices are formed by concatenating data column vectors, e.g. ${\bf X} = [{\bf x}_1, {\bf x}_2,...,{\bf x}_t]$ so that ${\bf X}\in{\mathbb R}^{k\times t}$, and ${\bf S} = [{\bf s}_1, {\bf s}_2,...,{\bf s}_t]$ so that ${\bf S}\in{\mathbb R}^{d\times t}$. In this notation, we introduce a time-centering operation $\delta$ such that, for example, time-centered stimuli are $\delta{\bf X}:= {\bf X}-\bar{\bf X}$ where $\bar {\bf X} := {\bf X}\frac 1t {\bf 1}{\bf 1}^\top$ and ${\bf 1}$ is a $t$ dimensional column vector whose elements are all 1's.

Our goal is to recover ${\bf S}$ from ${\bf X} = {\bf A}{\bf S}$, where ${\bf A}$ is unknown. We make the following two assumptions:
First, sources are nonnegative and decorrelated, $\frac 1t \delta{\bf S}\,\delta{\bf S}^\top = {\bf I}_d$. Note that while general ICA and NICA problems are stated with the independence assumption on the sources, for our purposes, it is sufficient that they are only decorrelated. Second, the mixing matrix, ${\bf A}  \in {\mathbb R}^{k\times d}$  $(k\geq d)$, is full-rank.

We propose that the source matrix, ${\bf S}$, can be recovered from ${\bf X}$ in the following two steps, also illustrated in Fig. \ref{Fig1}:
\begin{enumerate}[leftmargin=*]
\item \textbf{Generalized Prewhitening:} Transform ${\bf X}$ to ${\bf H} = {\bf F}{\bf X}$, where ${\bf F}$ is $l\times k$ with $l\geq d$, so that $\frac{1}{t}\delta{\bf H} \, \delta{\bf H}^\top$
has $d$ unit eigenvalues and $l-d$ zero eigenvalues. When $l=d$, ${\bf H}$ is whitened, otherwise channels of ${\bf H}$ are correlated. Such prewhitening is useful because it implies ${\bf H}^\top{\bf H} =  {\bf S}^\top{\bf S}$ according to the following theorem.
\begin{Th}\label{th2} If ${\bf F} \in {\mathbb R}^{l\times k} (l\geq k)$ is such that ${\bf H}={\bf F}{\bf X}$ obeys 
\begin{align}\label{hhbar}
\frac{1}{t}\delta{\bf H} \, \delta{\bf H}^\top = {\bf U}^{H}{\bf \Lambda}^{H}{{\bf U}^{H}}^\top, 
\end{align}
an eigenvalue decomposition with ${\bf \Lambda}^{H} = {\rm diag} \big(\underset{d}{\underbrace{1,\ldots,1}},\underset{l-d}{\underbrace{0,\ldots,0}} \big)$, then 
\begin{align}
{\bf H}^\top{\bf H} =  {\bf S}^\top{\bf S}.
\end{align}

\end{Th}
\begin{proof} 
To see why \eqref{hhbar} is sufficient, first note that $\frac{1}{t}\delta{\bf H} \,\delta{\bf H}^\top = \left({\bf FA}\right)\left({\bf FA}\right)^\top$. This follows from the definition of ${\bf H}$ and \eqref{white}. If \eqref{hhbar} holds, then 
\begin{align}\label{FA}
\left({\bf FA}\right)\left({\bf FA}\right)^\top =  {\bf U}^{H}{\bf \Lambda}^{H}{{\bf U}^{H}}^\top.
\end{align}
In turn, this is sufficient to prove that $\left({\bf FA}\right)^\top\left({\bf FA}\right)={\bf I}_d$. To see that, assume an SVD decomposition of $\left({\bf FA}\right) = {\bf U}^{FA}{\bf \Lambda}^{FA}{{\bf V}^{FA}}^\top$. \eqref{FA} implies that ${\bf \Lambda}^{FA}{{\bf \Lambda}^{FA}}^\top = {\bf \Lambda}^H$, i.e. that the $d$ diagonal elements of ${\bf \Lambda}^{FA} \in {\mathbb R}^{l\times d}$ are all 1's. Then, 
\begin{align}\label{orthFA}
\left({\bf FA}\right)^\top\left({\bf FA}\right) =  {\bf I}_d.
\end{align}  
This gives us the desired results ${\bf H}^\top{\bf H} =  {\bf S}^\top\left({\bf FA}\right)^\top\left({\bf FA}\right){\bf S} = {\bf S}^\top{\bf S}$.
\end{proof}

\begin{rmk} If $l>d$, the channels of ${\bf H}$ are correlated, except in the special case ${\bf U}^H = {\bf I}_d$. The whitening used in Plumbley's analysis  \citep{plumbley2002conditions} requires $l=d$.
\end{rmk}

\item \textbf{NSM:} Solve the following optimization problem:
\begin{align}\label{NSM}
{\bf Y}^* = \argmin_{{\bf Y},\, {\bf Y}\geq0} \left\| {\bf H}^\top{\bf H}-{\bf Y}^\top{\bf Y}\right\|_F^2,
\end{align}
where the optimization is performed over nonnegative ${\bf Y} := \left[{\bf y}_1,\ldots,{\bf y}_t\right]$ i.e. ${\bf y}_i \in {\mathbb R}_+^d$. We call \eqref{NSM} the NSM cost function \citep{pehlevan2014NMF}. 
Because inner products quantify similarities we call ${\bf H}^\top{\bf H}$ and ${\bf Y}^\top{\bf Y}$ input and output similarity matrices, i.e. their elements hold the pairwise similarities between input and the pairwise similarities between output vectors, respectively. Then, the cost function \eqref{NSM} preserves the input similarity structure as much as possible under the nonnegativity constraint. Variants of \eqref{NSM} has been considered in applied math literature under the name ``nonnegative symmetric matrix factorization'' for clustering applications, e.g. \citep{kuang2012symmetric,kuang2015symnmf}. 

Now we make our key observation. Using Theorem \ref{th2}, we can rewrite \eqref{NSM} as
\begin{align}
{\bf Y}^* =\argmin_{{\bf Y},\, {\bf Y}\geq0} \left\| {\bf S}^\top{\bf S}-{\bf Y}^\top{\bf Y}\right\|_F^2.
\end{align}
Since both ${\bf S}$ and ${\bf Y}$ are nonnegative, rank-$d$ matrices, ${\bf Y}^* = {\bf P}{\bf S}$, where ${\bf P}$ is a permutation matrix, is a solution to this optimization problem and the sources are successfully recovered. 

Uniqueness of the solutions (up to permutations) is hard to establish.   While both sufficient conditions, and necessary and sufficient conditions for uniqueness exist, these are non-trivial to verify and usually the verification is NP-complete \citep{donoho2003does,laurberg2008theorems,huang2014non}. A review of related uniqueness results can be found in \citep{huang2014non}. A necessary condition for uniqueness given in \citep{huang2014non} states that, if the factorization of ${\bf S}^\top{\bf S}$ to ${\bf Y}^\top{\bf Y}$ is unique (up to permutations), then each row of ${\bf S}$ contains at least one element that is equal to $0$. This necessary condition is similar to Plumbley's well-groundedness requirement used in proving Theorem 1.

The NSM problem can be solved by projected gradient descent, 
\begin{align}\label{pgd}
{\bf Y} \longleftarrow \max \left({\bf Y} + \eta \left({\bf Y}{\bf H}^\top{\bf H}-{\bf Y}{\bf Y}^\top{\bf Y}\right),0\right),
\end{align}
where the $\max$ operation is applied elementwise, and $\eta$ is the size of the gradient step. Other algorithms can be found in \citep{kuang2012symmetric,kuang2015symnmf,huang2014non}.
\end{enumerate}

\section{Derivation of NICA neural networks from similarity matching objectives}\label{streaming} \label{online}

Our analysis in the previous section revealed that the NICA problem can be solved in two steps: generalized prewhitening and nonnegative similarity matching. Here, we derive neural networks for each of these steps and stack them to give a biologically plausible two-layer neural network that operates in a streaming setting.

In a departure from the previous section, the number of output channels is reduced to the number of sources at the prewhitening stage rather than the later NSM stage ($l=d$). This assumption simplifies our analysis significantly. The full problem is addressed in Appendix \ref{FullGenWhite}.

\subsection{Noncentered prewhitening in a streaming input setting}

To derive a neurally plausible online algorithm for prewhitening, we pose generalized prewhitening, Eq. \eqref{hhbar}, as an optimization problem. Online minimization of this optimization problem gives an algorithm that can be mapped to the operation of a biologically plausible neural network.

Generalized prewhitening solves a constrained similarity alignment problem:
\begin{align}\label{NIPS3A}
\max_{\delta{\bf H}} {\rm Tr}\left(\delta {\bf X}^\top\delta{\bf X}\,\delta{\bf H}^\top \delta{\bf H}\right) \qquad {\rm s.t.} \quad     {\delta \bf H}^\top {\delta \bf H}\preceq t{\bf I}_t,
\end{align}
where $\delta {\bf X}$ is the $k\times t$ centered mixture of $d$ independent sources and $\delta{\bf H}$ is a $d\times t$ matrix, constrained such that $t{\bf I}_t-{\delta \bf H}^\top {\delta \bf H}$ is positive semidefinite. The solution of this objective aligns similarity matrices $\delta {\bf X}^\top\delta{\bf X}$ and $\delta{\bf H}^\top \delta{\bf H}$ so that their right singular vectors are the same \citep{pehlevan2015normative}. Then, the objective under the trace diagonalizes and its value is the sum of eigenvalue pair products. Since the eigenvalues of $\delta{\bf H}^\top \delta{\bf H}$ are upper bounded by $t$, the objective \eqref{NIPS3A} is maximized by setting the eigenvalues of $\delta{\bf H}^\top \delta{\bf H}$ that pair with the top $d$ eigenvalues of $\delta {\bf X}^\top\delta{\bf X}$ to $t$, and the rest to zero. Hence, the optimal $\delta {\bf H}$ satisfies the generalized prewhitening condition \eqref{hhbar}\citep{pehlevan2015normative}. More formally,
\begin{Th}[Modified from \citep{pehlevan2015normative}]Suppose an eigen-decomposition of $\delta{\bf X}^\top\delta{\bf X}$ is $\delta{\bf X}^\top\delta{\bf X} = {\bf V}^X {\bf \Lambda}^X {{\bf V}^X}^\top$, where eigenvalues are sorted in decreasing order.   Then, all optimal $\delta{\bf H}$ of \eqref{NIPS3A} have an SVD decomposition of the form
	\begin{align}\label{wYZ}
	\delta{\bf H}^* ={\bf U}^H \,\sqrt{t}\, {{\bf \Lambda}^H}' \, {{\bf V}^X}^\top,
	\end{align}
	where  ${{\bf \Lambda}^H}'$ is $d\times t$ with $d$ ones on top of the diagonal and zeros on the rest of the diagonal.
\end{Th}
The theorem implies that, first, $\frac{1}{t}\delta{\bf H}^* \, {\delta{\bf H}^*}^\top = {\bf I}_d$, and hence $\delta{\bf H}$ satisfies the generalized prewhitening condition \eqref{hhbar}. Second, ${\bf F}$, the linear mapping between $\delta{\bf H}^*$ and ${\delta\bf X}$, can be constructed using an SVD decomposition of $\delta{\bf X}$ and \eqref{wYZ}. 

The constraint in \eqref{NIPS3A} can be introduced into the objective function using as a Lagrange multiplier the Grassmanian of matrix $\delta{\bf G}\in {\mathbb R}^{m\times t}$ with ($m\geq d$): 
\begin{align}\label{NIPS3}
\min_{\delta{\bf H}} \max_{\delta{\bf G}} {\rm Tr}\left(-\delta {\bf X}^\top\delta{\bf X}\,\delta{\bf H}^\top \delta{\bf H} + \delta{\bf G}^\top \delta{\bf G}\left( \delta{\bf H}^\top \delta{\bf H} - t{\bf I}_t\right) \, \right),
\end{align}
This optimization problem \citep{pehlevan2015normative} will be used to derive an online neural algorithm.
%
	%
	%


Whereas the optimization problem \eqref{NIPS3} is formulated in the offline setting, i.e. outputs are computed only after receiving all inputs, to derive a biologically plausible algorithm, we need to formulate the optimization problem in the online setting, i.e. the algorithm receives inputs sequentially, one at a time, and computes an output before the next input arrives.  Therefore, we optimize \eqref{NIPS3} only for the data already received and only with respect to the current output: 
\begin{align}\label{whitenOnline}
\{\delta{\bf h}_t,\delta{\bf g}_t\} \longleftarrow 
\mathop {\arg \min }\limits_{\delta{\bf h}_t}  \mathop {\arg \max }\limits_{{\bf g}_t}{\rm Tr}\left(-\delta {\bf X}^\top\delta{\bf X}\,\delta{\bf H}^\top \delta{\bf H} + \delta{\bf H}^\top \delta{\bf H}\,\delta{\bf G}^\top \delta{\bf G} - t \, \delta{\bf G}^\top \delta{\bf G}\right).
 \end{align}
By keeping only those terms that depend on $\delta{\bf h}_t$ or $\delta{\bf g}_t$, we get the following objective:
\begin{align}\label{whitenOnline2}
\{\delta{\bf h}_t,\delta{\bf g}_t\} \longleftarrow \mathop {\arg \min }\limits_{\delta{\bf h}_t}  \mathop {\arg \max }\limits_{\delta{\bf g}_t} L(\delta{\bf h}_t,\delta{\bf g}_t),
 \end{align}
 where
\begin{align}\label{whitenOnlineReduced}
L =  - 2{\delta{\bf x}^\top_t}\left( {\sum\limits_{t' = 1}^{t - 1} {\delta{\bf x}_{t'}{\delta{\bf h}^\top_{t'}}} } \right)\delta{\bf h}_t -t\left\| \delta{{\bf g}_t} \right\|^2_2 &+ 2{\delta{\bf g}^\top_t}\left( {\sum\limits_{t' = 1}^{t - 1} {\delta{\bf g}_{t'}{\delta{\bf h}^\top_{t'}}} } \right)\delta{\bf h}_t \nonumber \\
&\qquad\qquad+\left(\left\| \delta{{\bf g}_t} \right\|^2_2 -\left\| \delta{{\bf x}_t} \right\|^2_2\right) \left\| \delta{{\bf h}_t} \right\|^2_2 .
 \end{align}
In the large-$t$ limit, the first three terms dominate over the last term, which we ignore. The remaining objective is strictly concave in $\delta{\bf g}_t$ and convex in $\delta{\bf h}_t$. We assume that the matrix $ \frac 1t {\sum\limits_{t' = 1}^{t - 1}  \delta{\bf h}_{t'}}{\delta{\bf g}^\top_{t'}} $ is full-rank. Then, the objective has a unique saddle point :
\begin{align}\label{onlinesaddle}
\delta{\bf g}^*_t &= {\bf W}^{GH}_t\delta{\bf h}^*_t, \nonumber \\
\delta{\bf h}^*_{t} &= \left({\bf W}^{HG}_t{\bf W}^{GH}_t\right)^{-1}{\bf W}^{HX}_t\delta{\bf x}_t,
\end{align}
where,
\begin{align}\label{weightdef}
{\bf W}^{HX}_t &:= \frac 1t {\sum\limits_{t' = 1}^{t - 1} {\delta{\bf h}_{t'}}  \delta{\bf x}_{t'}^\top}, \qquad {\bf W}^{HG}_t :=  \frac 1t {\sum\limits_{t' = 1}^{t - 1} {\delta{\bf h}_{t'}}  \delta{\bf g}_{t'}^\top}, \nonumber \\
{\bf W}^{GH}_t &:= {{\bf W}^{HG}_t}^\top :=\frac 1t {\sum\limits_{t' = 1}^{t' - 1} {\delta{\bf g}_{t'}}  \delta{\bf h}_{t'}^\top}.
 \end{align}
Hence, ${\bf F}_t:=\left({\bf W}^{HG}_t{\bf W}^{GH}_t\right)^{-1}{\bf W}^{HX}_t$ can be interpreted as the current estimate of the prewhitening matrix, ${\bf F}$.

We solve \eqref{whitenOnline2} with a gradient descent-ascent
\begin{align}\label{dynamicsdef}
\frac{d\delta{\bf h}_t}{d\gamma} &= -\frac{1}{2t}\nabla_{{\delta{\bf h}_t}}{L} ={\bf W}^{HX}_t\delta{\bf x}_t - {\bf W}^{HG}_t{\bf g}_t, \nonumber \\
\frac{d\delta{\bf g}_t}{d\gamma} &= \frac{1}{2t}\nabla_{{\delta{\bf g}_t}}{L} = - \delta{\bf g}_t + {\bf W}^{GH}_t\delta{\bf h}_t.
\end{align}
where $\gamma$ measures ``time'' within a single time step of $t$. Biologically, this is justified if the activity dynamics converges faster than the correlation time of the input data. The dynamics \eqref{dynamicsdef} can be proved to converge to the saddle point of the objective \eqref{whitenOnlineReduced}, see Appendix \ref{app_convergence}. 

Equation \eqref{dynamicsdef} describes the dynamics of a single-layer neural network with two-populations, Fig. \ref{Fig2}. ${\bf W}^{HX}_t$ represents the weights of feedforward  synaptic connections, ${\bf W}^{HG}_t$ and ${\bf W}^{GH}_t$ represent the weights of synaptic connections between the two populations. Remarkably, synaptic weights appear in the online algorithm despite their absence in the optimization problem formulations \eqref{NIPS3} and \eqref{whitenOnline}. Furthermore, $\delta\bf{h}_t$ neurons can be associated with principal neurons of a biological circuit and $\delta\bf{g}_t$ neurons with interneurons. 

Finally, using the definition of synaptic weight matrices \eqref{weightdef}, we can formulate recursive update rules:
\begin{align}\label{updatedef}
{\bf W}^{HX}_{t+1} &= {\bf W}^{HX}_{t} +    \frac{1}{t+1} \left(\delta{\bf h}_{t}  \delta{\bf x}_{t}^\top - {\bf W}^{HX}_{t} \right), \nonumber \\ 
{\bf W}^{HG}_{t+1} &= {\bf W}^{HG}_{t} +    \frac{1}{t+1} \left(\delta{\bf h}_{t}  \delta{\bf g}_{t}^\top - {\bf W}^{HG}_{t} \right), \nonumber \\
{\bf W}^{GH}_{t+1} &= {\bf W}^{GH}_{t} +    \frac{1}{t+1} \left(\delta{\bf g}_{t}  \delta{\bf h}_{t}^\top - {\bf W}^{GH}_{t} \right).
\end{align}

Equations \eqref{dynamicsdef} and \eqref{updatedef} define a neural algorithm that proceeds in two phases. After each
stimulus presentation, first, \eqref{dynamicsdef} is iterated until convergence by the dynamics of neuronal activities. Second, synaptic weights are updated according
to local, anti-Hebbian (for synapses from interneurons) and Hebbian (for all other synapses) rules \eqref{updatedef}.  Biologically, synaptic weights are updated on a slower timescale than neuronal activity dynamics.

Our algorithm can be viewed as a special case of the algorithm proposed in \citep{plumbley1996information,plumbley1994subspace}. Plumbley analyzed the convergence of synaptic weights \citep{plumbley1994subspace} in a stochastic setting by a linear stability analysis of the stationary point of synaptic weight updates. His results are directly applicable to our algorithm, and show that, if the synaptic weights of our algorithm converge to a stationary state, they whiten the input. 

Importantly, unlike \citep{plumbley1996information,plumbley1994subspace} which proposed the algorithm heuristically, we derived it by posing and solving an optimization problem.

\begin{figure}
\centering
\includegraphics{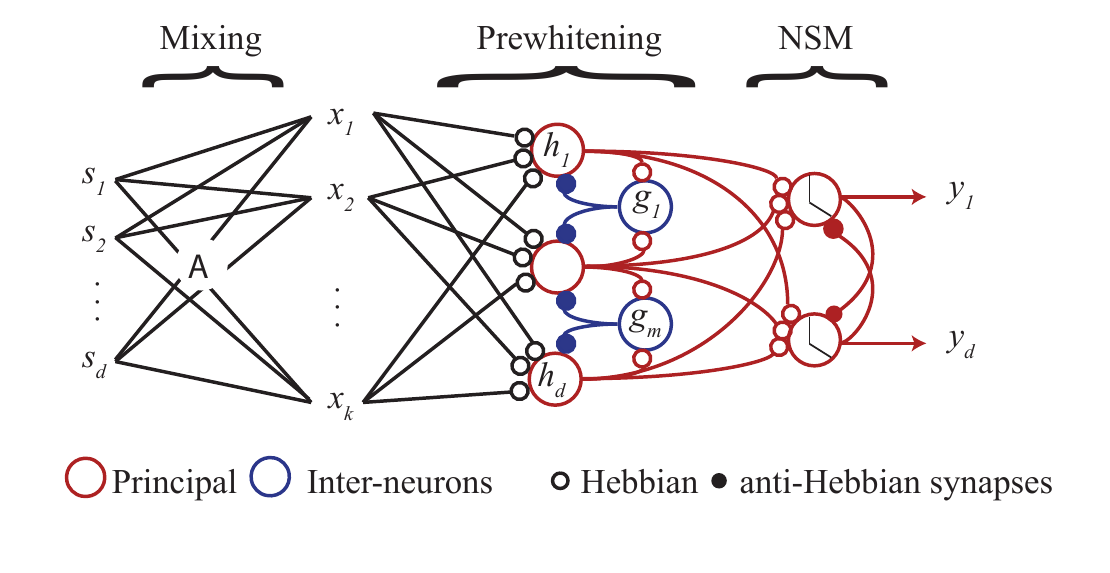}
\caption{Biologically plausible network for blind source separation. The prewhitening stage is composed of two populations of linear neurons. The NSM stage has a single population of rectifying neurons. \label{Fig2}}
\end{figure}

\subsubsection{Computing $\bar{\bf H}$} 

The optimization problem \eqref{NIPS3} and the corresponding neural algorithm, Eqs. \eqref{dynamicsdef} and \eqref{updatedef} almost achieve what is needed for noncentered prewhitening, but we still need to find $\bar {\bf H}$, since for the NSM step we need ${\bf H} = \delta{\bf H} + \bar{\bf H}$. We now discuss how $\bar{\bf H}$ can be learned along with $\delta{\bf H}$ using the same online algorithm. 

Our online algorithm for centered-whitening is of the following form. First, a neural dynamics stage outputs a linear transformation of the input:
\begin{align}\label{upd1}
\delta {\bf h}_t= {\bf F}_{t}\delta {\bf x}_t, 
\end{align}
and, second, synaptic weights and, hence,  ${\bf F}_t$ are updated: 
\begin{align}\label{upd3}
{\bf F}_{t+1} = {\bf F}_{t} +\delta {\bf F}_t (t,\delta {\bf h}_t, \delta {\bf g}_t, {\bf F}_{t}).
\end{align}
We can supplement this algorithm with a running estimate of $\bar{\bf h}$. Let the running estimate of average stimulus activity be
\begin{align}\label{upd6}
\bar {\bf x}_t=\frac 1t\sum_{t'=1}^t {\bf x}_{t'} = \left(1-\frac 1t\right) \bar{\bf x}_{t-1} + \frac 1t {\bf x}_t. 
\end{align}
Then,
\begin{align}\label{upd2}
\bar{\bf h}_{t} &= {\bf F}_{t} \bar{\bf x}_t. 
\end{align}

Alternatively, \eqref{upd1} and \eqref{upd2} can be combined into a single step: 
\begin{align}\label{upd4}
{\bf h}_t = {\bf F}_{t} {\bf x}_t,
\end{align}
where the network receives uncentered stimuli and prewhitenes it. Note that assignment \eqref{upd4} can still be done by iterating \eqref{dynamicsdef}, except now the input is ${\bf x}_t$ rather than $\delta{\bf x}_t$.
However, synaptic weights are still updated using $\delta{\bf x}_t = {\bf x}_t - \bar{\bf x}_t$, $\delta{\bf h}_t = {\bf h}_t - \bar{\bf h}_t$ and $\delta{\bf g}_t = {\bf g}_t - \bar{\bf g}_t$. Therefore we keep recursive estimates of the means. Substituting \eqref{upd6} into \eqref{upd4} and using \eqref{upd3}
\begin{align}\label{upd7}
\bar{\bf h}_{t} = \left(1-\frac 1t\right)\bar{\bf h}_{t-1} + \left(1-\frac 1t\right) \delta{\bf F}_{t-1}\bar{\bf x}_{t-1}  + \frac 1t {\bf h}_t. 
\end{align}
The term $ \left(1-\frac 1t\right) \delta{\bf F}_{t-1}\bar{\bf x}_{t-1}$ requires non-local calculations. Assuming that in the large-$t$ limit updates to ${\bf F}$ are small, we can ignore this term and obtain a recursion:
\begin{align}\label{upd8}
\bar{\bf h}_{t} =  \left(1-\frac 1t\right)\bar{\bf h}_{t-1} + \frac 1t {\bf h}_t. 
\end{align}
Finally, a similar argument can be given for $\bar{\bf g}_{t}$. We keep a recursive estimate of $\bar{\bf g}_{t}$:
\begin{align}\label{upd9}
\bar{\bf g}_{t} = \left(1-\frac 1t\right) \bar{\bf g}_{t-1} + \frac 1t {\bf g}_t.
\end{align}
%

To summarize, when a new stimulus ${\bf x}_t$ is observed, the algorithm operates in two steps. In the first step, the following two-population neural dynamics runs until convergence to a fixed point:
\begin{align}\label{fullFirst}
\frac{d{\bf h}_t}{d\gamma} &= {\bf W}^{HX}_t{\bf x}_t - {\bf W}^{HG}_t{\bf g}_t, \nonumber \\
\frac{d{\bf g}_t}{d\gamma} &=  - {\bf g}_t + {\bf W}^{GH}_t{\bf h}_t,
\end{align}
The convergence proof  for neural dynamics \eqref{dynamicsdef} given in Appendix \ref{app_convergence} also applies here. Besides the synaptic weight, each neuron remembers its own average activity and each synapse remembers average incoming activity. In the second step of the algorithm, the average activities are updated by:
\begin{align}\label{fullFirstbar}
\bar{\bf x}_{t} &= \left(1-\frac 1t\right) \bar{\bf x}_{t-1} + \frac 1t {\bf x}_t,\nonumber \\
\bar{\bf h}_{t} &=  \left(1-\frac 1t\right)\bar{\bf h}_{t-1} + \frac 1t {\bf h}_t, \nonumber \\
\bar{\bf g}_{t} &= \left(1-\frac 1t\right) \bar{\bf g}_{t-1} + \frac 1t {\bf g}_t.
\end{align}
Synaptic weight matrices are updated recursively by 
\begin{align}\label{fullFirstW}
{\bf W}^{HX}_{t+1} &= {\bf W}^{HX}_{t} +    \frac{1}{t+1} \left(\delta{\bf h}_{t} \delta{\bf x}_{t}^\top - {\bf W}^{HX}_{t} \right), \nonumber \\ 
{\bf W}^{HG}_{t+1} &= {\bf W}^{HG}_{t} +    \frac{1}{t+1} \left(\delta{\bf h}_{t}  \delta{\bf g}_{t}^\top - {\bf W}^{HG}_{t} \right), \nonumber \\
{\bf W}^{GH}_{t+1} &= {\bf W}^{GH}_{t} +    \frac{1}{t+1} \left(\delta{\bf g}_{t}  \delta{\bf h}_{t}^\top - {\bf W}^{GH}_{t} \right).
\end{align}
Once the synaptic updates are done, the new stimulus, ${\bf x}_{t+1}$, is processed. We note again that all the synaptic update rules are local, and hence are biologically plausible.

\subsection{Online NSM}

Next, we derive the second-layer network which solves the NSM optimization problem \eqref{NSM} in an online setting \citep{pehlevan2014NMF}.

The online optimization problem is:
\begin{align}\label{NSMonline}
{\bf y}_t \longleftarrow \argmin_{{\bf y}_t,\, {\bf y}_t\geq0} \left\| {\bf H}^\top{\bf H}-{\bf Y}^\top{\bf Y}\right\|_F^2.
\end{align}
Proceeding as before, let's rewrite \eqref{NSMonline} keeping only terms that depend on ${\bf y}_t$:
\begin{align}\label{NSMonline2}
{\bf y}_t \longleftarrow \argmin_{{\bf y}_t,\, {\bf y}_t\geq0}\left( 2 {\bf y}_t^\top \left(\sum_{t'=1}^{t-1}{\bf y}_{t'}{\bf y}_{t'}^\top\right){\bf y}_t - 4   {\bf h}_t^\top \left(\sum_{t'=1}^{t-1}{\bf h}_{t'}{\bf y}_{t'}^\top\right){\bf y}_t+ \left\| {\bf y}_t \right\|^4_2 - 2 \left\| {\bf h}_t \right\|^2_2 \left\| {\bf y}_t \right\|^2_2\right).
\end{align}
In the large-$t$ limit, the last two terms can be ignored and the remainder is a quadratic form in ${\bf y}_t$. We minimize it using coordinate descent  \citep{wright2015coordinate} which is both fast and neurally plausible. In this approach, neurons are updated one-by-one by performing an exact minimization of \eqref{NSMonline2} with respect to $y_{t,i}$ until convergence: 
\begin{align}\label{NSMdynamics}
y_{t,i}\longleftarrow \max\left(\sum_{j}{W}^{YH}_{t,ij}{h}_{t,j}-{W}_{t,ij}^{YY}{y}_{t,j},0\right),
\end{align}
where
\begin{align}\label{WNSMdef}
{W}_{t,ij}^{YH} = \frac{\sum_{t'=1}^{t-1}y_{t',i}h_{t',j}}{\sum_{t'=1}^{t-1}y_{t',i}^2}, \qquad  {W}_{t,i,j\neq i}^{YY} = \frac{\sum_{t'=1}^{t-1}y_{t',i}y_{t',j}}{\sum_{t'=1}^{t-1}y_{t',i}^2}, \qquad {W}_{t,ii}^{YY}=0.
\end{align}

For the next time step ($t+1$), we can update the synaptic weights recursively, giving us the following local learning rules:
\begin{align}\label{fullNSMW}
D_{t+1,i} &= D_{t,i} + y_{t,i}^2, \nonumber \\
{W}_{t+1,ij}^{YH} &= {W}_{t,ij}^{YH} + \frac{1}{D_{t+1,i}}\left(y_{t,i}h_{t,j}-y_{t,i}^2W^{YH}_{t,ij}\right), \nonumber \\
{W}_{t+1,i,j\neq i}^{YY} &= {W}_{t,ij}^{YY} + \frac{1}{D_{t+1,i}}\left(y_{t,i}y_{t,j}-y_{t,i}^2W^{YY}_{t,ij}\right), \qquad  {W}_{t,ii}^{YY}=0.
\end{align}
Interestingly, these update rules are local and are identical to the single-neuron Oja rule \citep{oja1982simplified}, except that the learning rate is given explicitly in terms of cumulative activity $1/D_{t,i}$ and the lateral connections are anti-Hebbian.

After the arrival of each data vector, the operation of the complete two-layer network algorithm, Fig. \ref{Fig2}, is as follows. First, the dynamics of the prewhitening network runs until convergence. Then the output of the prewhitening network is fed to the NSM network, and the NSM network dynamics runs until convergence to a fixed point. Synaptic weights are updated in both networks for processing the next data vector.

\subsubsection{NICA is a stationary state of online NSM}\label{convergence2}

Here we show that the solution to the NICA problem is a stationary synaptic weights state of the online NSM algorithm. In the stationary state the expected updates to synaptic weights are zero, i.e.
\begin{align}\label{stationaryDef}
\left\langle \Delta W^{YH}_{ij}\right\rangle = 0, \qquad \left\langle \Delta W^{YY}_{ij}\right\rangle = 0, 
\end{align}
where we dropped the $t$ index, and brackets denote averages over the source distribution. 

Suppose the stimuli obey the NICA generative model, Eq. \eqref{genmod}, and the observed mixture, ${\bf x}_t$, is whitened with the exact (generalized) prewhitening matrix ${\bf F}$ described in Theorem \ref{th2}.  Then, input to the network at time, $t$, is  ${\bf h}_t = {\bf F}{\bf x}_t = {\bf F}{\bf A} {\bf s}_t$. Our claim is that there exists synaptic weight configurations for which 1) for any mixed input, ${\bf x}_t$, the output of the network is the source vector, i.e. ${\bf y}_t = {\bf P}{\bf s}_t$, where ${\bf P}$ is a permutation matrix, and 2) this synaptic configuration is a stationary state. 

We prove our claim by constructing these synaptic weights. For each permutation matrix, we first relabel the outputs such that $i^{\rm th}$ output recovers the $i^{\rm th}$ source and hence ${\bf P}$ becomes the identity matrix. Then, the weights are:
\begin{align}\label{stationaryWeights}
W^{YY}_{ij} = \frac{\left\langle s_i\right\rangle\left\langle s_j\right\rangle}{\left\langle s_i^2\right\rangle}, \qquad W^{YH}_{ij} = \frac{\left\langle s_ih_j\right\rangle}{\left\langle s_i^2\right\rangle}=\sum_{k\neq i}({\bf FA})_{jk}\frac{\left\langle s_k \right\rangle \left\langle s_i\right\rangle}{\left\langle s_i^2\right\rangle} + ({\bf FA})_{ji}.
\end{align}
Given mixture ${\bf x}_t$, NSM neural dynamics with these weights converge to $y_{t,i} = s_{t,i}$, which is the the unique fixed point\footnote{\textbf{Proof:} The net input to neuron $i$ at the claimed fixed point is $\sum_{j}W^{YH}_{ij} h_{t,j} - \sum_{j\neq i}W^{YY}_{ij} s_{t,j}$. Plugging in \eqref{stationaryWeights} and ${\bf h}_t =  {\bf F}{\bf A} {\bf s}_t$, and  using \eqref{orthFA} one gets that the net input is $s_{t,i}$, which is also the output since sources are nonnegative. This fixed point is unique and globally stable because the NSM neural dynamics is a coordinate descent on a strictly convex cost given by $\frac 12{\bf y}_t^\top \left\langle {\bf s} {\bf s}^\top\right\rangle{\bf y}_t - {\bf h}_t^\top\left\langle {\bf h} {\bf s}^\top\right\rangle{\bf y}_t$.}. Furthermore, these weights define a stationary state as defined in \eqref{stationaryDef} assuming a fixed learning rate. To see this substitute weights from \eqref{stationaryWeights} into the last two equations of \eqref{fullNSMW} and average over the source distribution. The fixed learning rate assumption is valid in the large-$t$ limit when changes to $D_{t,i}$ become small (${\mathcal{O}(1/t)}$, see \citep{pehlevan2015MDS}).

%
%

\section{Numerical simulations}\label{numerical}

Here we present numerical simulations of our two-layered neural network using various datasets and compare the results to that of other algorithms.

In all our simulations, $d=k=l=m$, except in Fig. \ref{Fig5}B where $d=k>l=m$. Our networks were initialized as follows:
\begin{enumerate}
\item In the prewhitening network, ${\bf W}^{HX}$ and ${\bf W}^{HG}$ were chosen to be random orthonormal matrices. ${\bf W}^{GH}$ is initialized as ${{\bf W}^{HG}}^\top$ because of its definition in Eq. \eqref{weightdef} and the fact that this choice guarantees the convergence of the neural dynamics \eqref{fullFirst} (see Appendix \ref{app_convergence}).

\item In the NSM network, ${\bf W}^{YH}$ was initialized to a random orthonormal matrix and ${\bf W}^{YY}$ was set to zero.
\end{enumerate}
The learning rates were chosen as follows:
\begin{enumerate}
\item For the prewhitening network, we generalized the time-dependent learning rate  \eqref{fullFirstW} to,
\begin{align}\label{timedep}
\frac{1}{a+bt},
\end{align}
and performed a grid search over $a\in\lbrace 10, 10^2, 10^3, 10^4\rbrace$ and $b\in \lbrace 10^{-2}, 10^{-1}, 1\rbrace$ to find the combination with best performance. Our performance measures will be introduced below.

\item For the NSM network, we generalized the activity-dependent learning rate \eqref{fullNSMW} to,
\begin{align}\label{activitydep}
\frac{1}{\tilde D_{t+1,i}}, \qquad {\rm where} \qquad \tilde D_{t+1,i} = \min \left(\tilde a, \tilde b\tilde D_{t+1,i} +  y_{t,i}^2 \right), 
\end{align}
and performed a grid search over several values of $\tilde a\in\lbrace 10, 10^2, 10^3, 10^4\rbrace$ and $\tilde b\in \lbrace 0.8, 0.9, 0.95, 0.99, 0.995, 0.999, 0.9999, 1\rbrace$ to find the combination with best performance. The $\tilde b$ parameter introduces ``forgetting'' to the system \citep{pehlevan2015MDS}. We hypothesized that forgetting will be beneficial in the two-layer setting because the prewhitening layer output changes over time and the NSM layer has to adapt. Further, for comparison purposes, we also implemented this algorithm with a time-dependent learning rate of the form \eqref{timedep} and performed a grid search with $a\in\lbrace 10^2, 10^3, 10^4\rbrace$ and $b\in \lbrace 10^{-2}, 10^{-1}, 1\rbrace$ to find the combination with best performance. 
\end{enumerate}
For the NSM network, to speed up our simulations we implemented a procedure from \citep{plumbley2004nonnegative}. At each iteration  we checked whether there is any output neuron who has not fired up until that iteration. If so, we flipped the sign of its feedforward inputs. In practice, the flipping only occured within the first $\sim$10 iterations. 

For comparison, we implemented five other algorithms. First is the offline algorithm  \eqref{pgd}, the other two are chosen to represent major online algorithm classes:
\begin{enumerate}
\item{\bf Offline projected gradient descent:} We simulated the projected gradient descent algorithm \eqref{pgd}. We used variable stepsizes of the form \eqref{timedep} and performed a grid search with $a\in\lbrace 10^4, 10^5, 10^6\rbrace$ and $b\in \lbrace 10^{-3}, 10^{-2}, 10^{-1}\rbrace$ to find the combination with best performance. We initialized elements of the  matrix, ${\bf Y}$, by drawing a Gaussian random variable with zero mean and unit variance and rectifying it. Input dataset was whitened offline before passing to projected gradient descent.
\item{\bf fastICA}: fastICA \citep{hyvarinen1997fast, hyvarinen1999fast, hyvarinen2000independent} is a popular ICA algorithm which does not assume nonnegativity of sources. We implemented an online version of fastICA  \citep{hyvarinen1998independent} using the same parameters except for feedforward weights. We used the time-dependent learning rate \eqref{timedep} and performed a grid search with $a\in\lbrace 10, 10^2, 10^3, 10^4\rbrace$ and $b\in \lbrace 10^{-2}, 10^{-1}, 1\rbrace$ to find the combination with best performance. fastICA requires whitened and centered input \citep{hyvarinen1998independent} and computes a decoding matrix that maps mixtures back to sources. We ran the algorithm with whitened and centered input. To recover nonnegative sources, we applied the decoding matrix to noncentered but whitened input.

\item{\bf Infomax ICA:} \cite{bell1995information} proposed a blind source separation algorithm that  maximizes the mutual information between inputs and outputs, namely the Infomax principle \citep{linsker1988self}. We simulated an online version due to \cite{cichocki1996new}. We chose cubic neural nonlinearities compatible with sub-Gaussian input sources. This differs from our fastICA implementation where the nonlinearity is also learned online. Infomax ICA computes a decoding matrix using centered, but not whitened, data. To recover nonnegative sources, we applied the decoding matrix to noncentered inputs. Finally, we rescaled the sources so that their variance is 1. We experimented with several learning rate parameters for finding optimal performance.
\item{\bf Linsker's network:} \cite{linsker1997local} proposed a neural network with local learning rules for Infomax ICA. We simulated this algorithm with cubic neural nonlinearities and preprocessing and decoding done as in our Infomax ICA implementation. 
\item {\bf Nonnegative PCA}: Nonnegative PCA algorithm \citep{plumbley2004nonnegative}  solves the NICA task and makes explicit use of the nonnegativity of sources. We use the online version given in \citep{plumbley2004nonnegative}. To speed up our simulations we implemented a procedure from \citep{plumbley2004nonnegative}. At each iteration  we checked whether there is any output neuron who has not fired up until that iteration. If so, we flipped the sign of its feedforward inputs. For this algorithm, we again used the time-dependent learning rate of \eqref{timedep} and performed a grid search with $a\in\lbrace 10, 10^2, 10^3, 10^4\rbrace$ and $b\in \lbrace 10^{-2}, 10^{-1}, 1\rbrace$ to find the combination with best performance. Nonnegative PCA assumes whitened, but not centered input \citep{plumbley2004nonnegative}.
\end{enumerate}

Next, we present the results of our simulations on three datasets.

\subsection{Mixture of random uniform sources}

\begin{figure}

\includegraphics{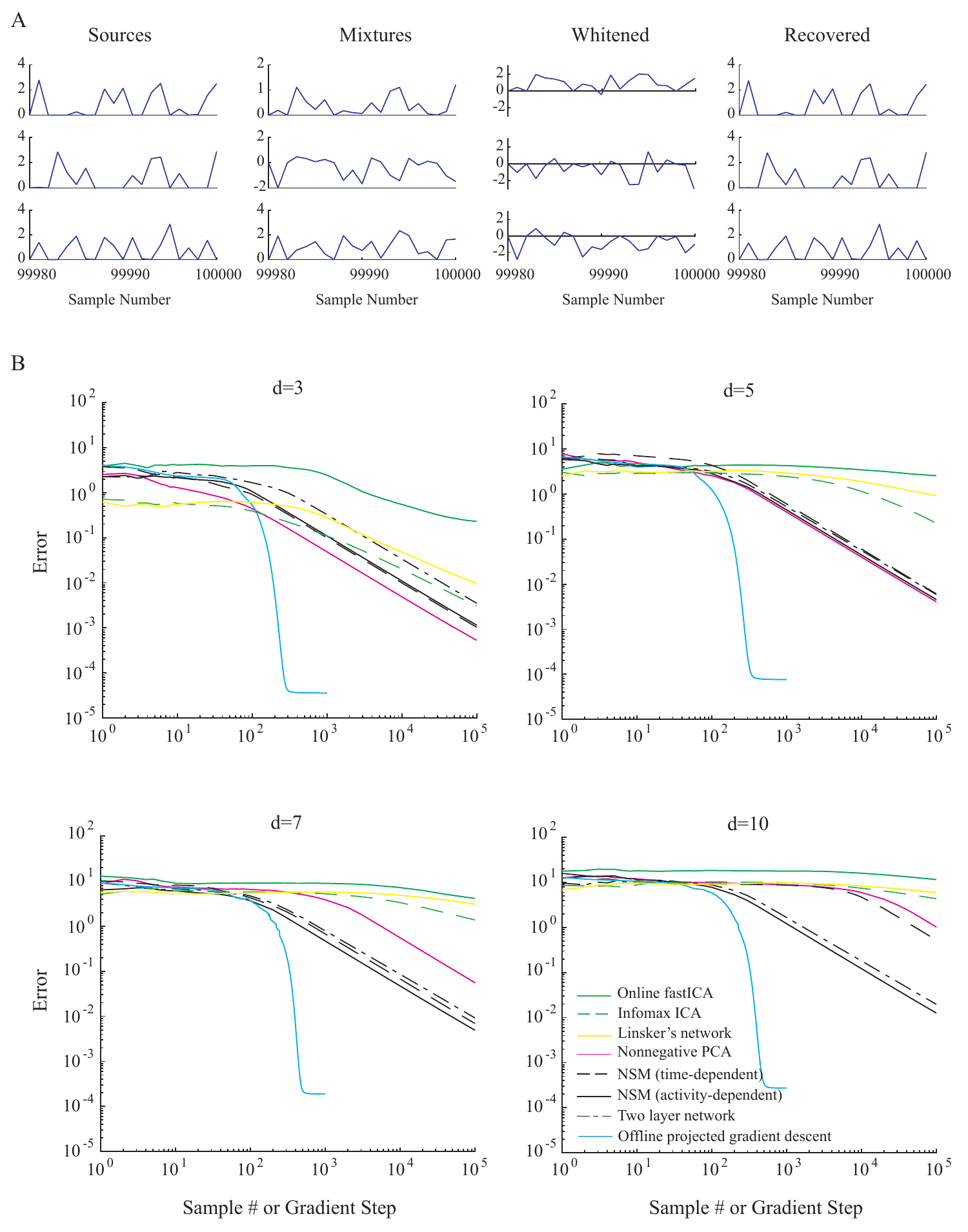}
\caption{Performance of our algorithms when presented with a mixture of random uniform sources: A) Sample source, mixture, whitened and recovered signals for a $d=3$ task, performed by our two-layer algorithm. The whitened signal is the output of the first layer, while the recovered signal is the output of second layer. B) Performance comparison of the online algorithms presented in this paper with projected gradient descent, Online fastICA, Infomax ICA, Linsker's network and Nonnegative PCA. Curves show averages over 10 simulations. Error bars not shown for visual clarity. Learning rate parameters are given in Appendix \ref{lrate}.\label{Fig3}}
\end{figure}

The source i.i.d. samples were set to zero with probability 0.5 and sampled uniformly from iterval $[0,\sqrt{48/5}]$ with probability 0.5. The dimensions of source vectors were $d=\lbrace3,5,7,10\rbrace$. The mixing matrices are given in Appendix \ref{mixing}. $10^5$ source vectors were generated for each run. 
For a sample of the original and mixed signals, see Fig \ref{Fig3}A.

The inputs to fastICA and Nonnegative PCA algorithms were prewhitened offline, and in the case of fastICA they were also centered. We ran our NSM network both as a single layer algorithm with prewhitening done offline, and as a part of our two-layer algorithm with whitening done online. 

To quantify the performance of tested algorithms, we used the mean-squared-error:
\begin{align}
E_{t} = \frac{1}{t}\sum_{t' = 1}^t\left \lVert {\bf s}_{t'} - {\bf P}{\bf y}_{t'} \right \rVert_2^2,
\end{align}
where ${\bf P}$ is a permutation matrix that is chosen to minimize the mean-squared-error at $t=10^5$. The learning rate parameters of all networks were optimized by a grid search using $E_{10^5}$ as the performance metric. 

In Fig. $\ref{Fig3}$B, we show the performances of all algorithms we implemented. Our algorithms perform as well or better than others, especially as dimensionality of the input increases. Offline whitening is better than online whitening, however, as dimensionality increases, online whitening becomes competitive with offline whitening. In fact, our two-layer and single-layer networks perform better than Online fastICA and Nonnegative PCA for which whitening was done offline. 

We also simulated a fully offline algorithm by taking projected gradient descent steps until the residual error plateaued (Fig. $\ref{Fig3}$B). The performance of the offline algorithm quantifies two important metrics. First, it establishes the loss in performance due to online (as opposed to offline) processing. Second, it establishes the lowest error that could be achieved by the NSM method for the given dataset. The lowest error is not necessarily zero due to the finite size of the dataset. This method is not perfect because the projected gradient descent may get stuck in a local minimum of Eq. \eqref{NSM}.

We also tested whether the learned synaptic weights of our network match our theoretical predictions. In Fig. \ref{Fig4}A, we show examples of learned feedforward and recurrent synaptic weights at $t=10^5$, and what is expected from our theory \eqref{stationaryWeights}. We observed an almost perfect match between the two. In Fig. \ref{Fig4}B, we quantify the convergence of simulated synaptic weights to the theoretical prediction by plotting a normalized error metric defined by $E_t = \left\Vert {\bf W}_{t,{\rm simulation}} -{\bf W}_{{\rm theory}}  \right\Vert_F^2/ \left\Vert {\bf W}_{{\rm theory}}  \right\Vert_F^2$.

\begin{figure}
\centering
\includegraphics{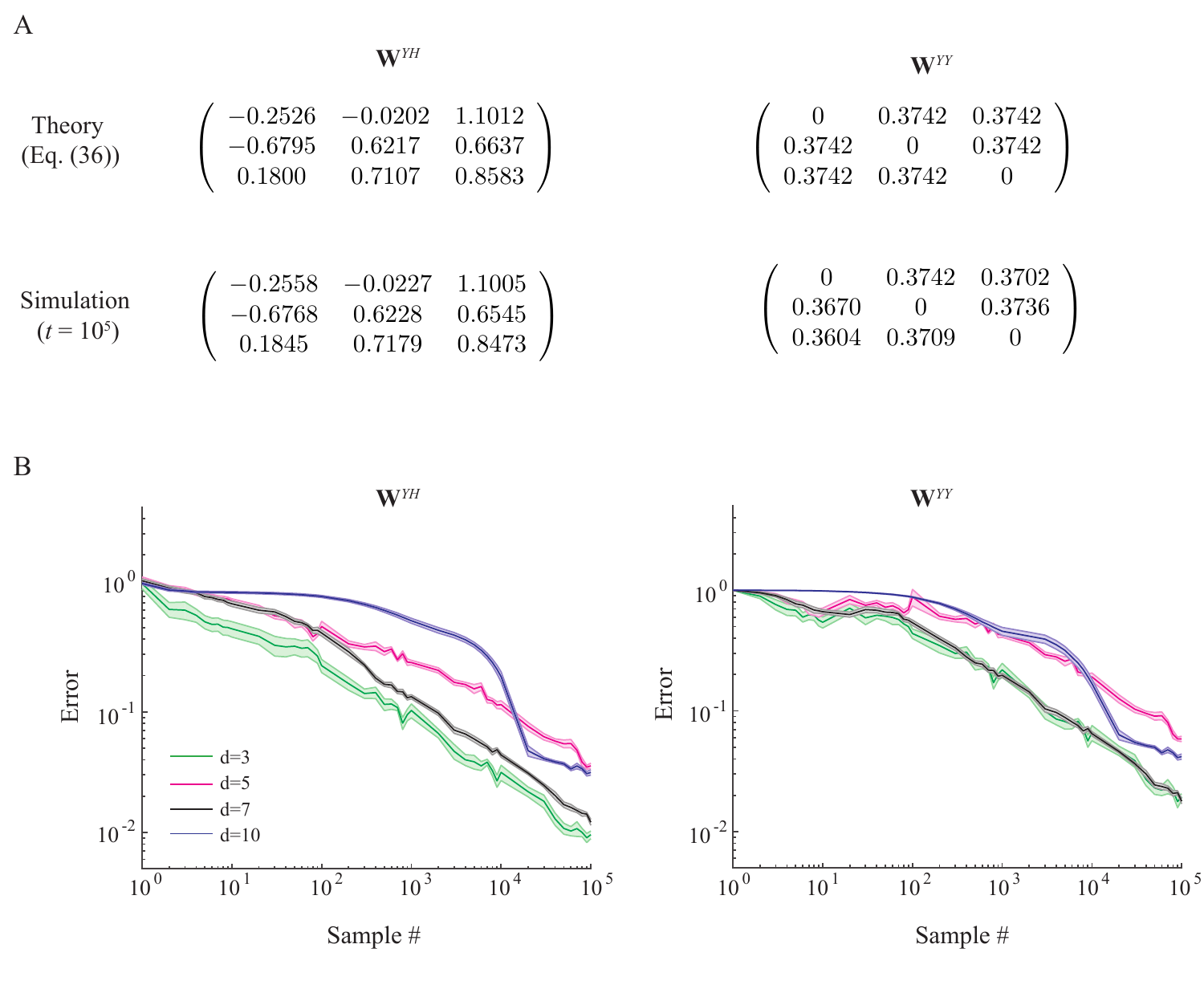}
\caption{Theoretical predictions of learned synaptic weights match simulations: A) Example synaptic weight matrices predicted from theory \eqref{stationaryWeights} compared to results from an example simulation ($t=10^5$). B) Convergence of simulated network synaptic weights to theoretical predictions. For this figure, inputs were whitened offline and the NSM network was run with time-dependent learning rates. Shaded bars show standard error over 10 simulations. \label{Fig4}}
\end{figure}

\subsection{Mixture of random uniform and exponential sources}

Our algorithm can demix sources sampled from different statistical distributions. To illustrate this point, we generated a dataset with two uniform and three exponential source channels. The uniform sources were sampled as before. The exponential sources were either zero (with probability 0.5) or sampled from an exponential distribution, scaled so that the variance of the channel is 1. In Fig. \ref{Fig5}A, we show that the algorithm succesfully recovers sources.

To test denoising capabilities of our algorithm, we created a dataset where source signals are accompanied by background noise. Sources to be recovered were three exponential channels, which were sampled as before. Background noises were two uniform channels which were sampled as before, except scaled to have variance 0.1. To denoise the resulting five dimensional mixture, the prewhitening layer reduced its five input dimensions to three. Then, the NSM layer succesfully recovered sources, Fig. \ref{Fig5}B. Hence, the prewhitening layer can act as a denoising stage.

\begin{figure}
	
	\includegraphics{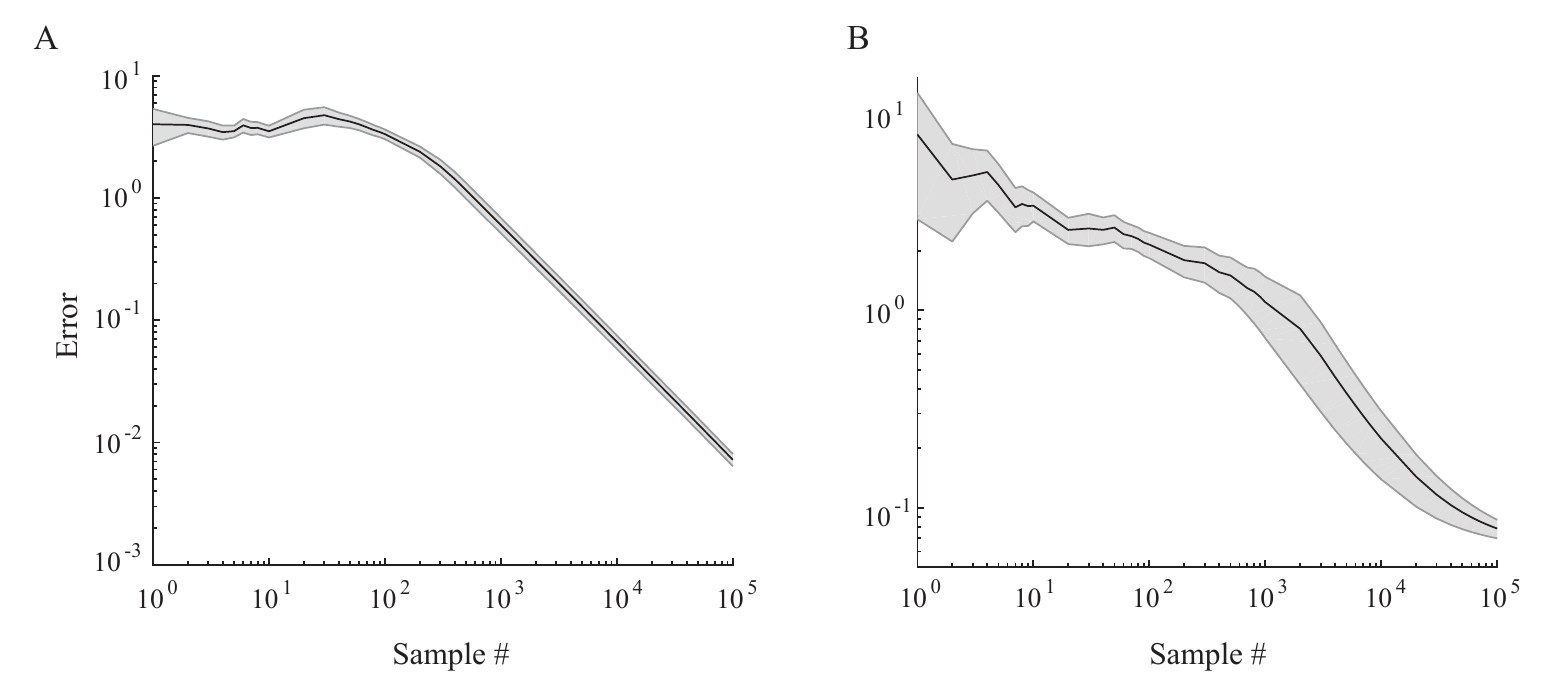}
	\caption{Performance of our two layer algorithm when presented with a mixture of random uniform and exponential sources. A) Recovery of a mixture of three exponential and two uniform sources.  B) Recovery of three exponential sources corrupted by two background noise channels. Learning rate parameters are given in Appendix \ref{lrate}.\label{Fig5}}
\end{figure}

\subsection{Mixture of natural scenes}

Next, we consider recovering images from their mixtures, Fig. \ref{Fig6}A, where each image is treated as one source. Four image patches of size $252\times 252$ pixels were chosen from a set of images of natural scenes which were previously used in \citep{hyvarinen2000emergence,plumbley2004nonnegative}. The preprocessing was as in \citep{plumbley2004nonnegative}: 1) Images were downsampled by a factor of 4 to obtain $63\times 63$ patches, 2) Pixel intensities were shifted to have a minimum of zero and 3) Pixel intensities were scaled to have unit variance. Hence, in this dataset, there are $d=4$ sources, corresponding image patches, and a total of $63\times63 = 3969$ samples. These samples were presented to the algorithm 5000 times with randomly permuted order in each presentation. The $4\times 4$ mixing matrix, which was generated randomly, is given in Appendix \ref{mixing}. 

In Fig. $\ref{Fig6}$B, we show the performances of all algorithms we implemented in this task. We see that our algorithms, when compared to fastICA and Nonnegative PCA, perform much better.

\begin{figure}
\centering
\includegraphics{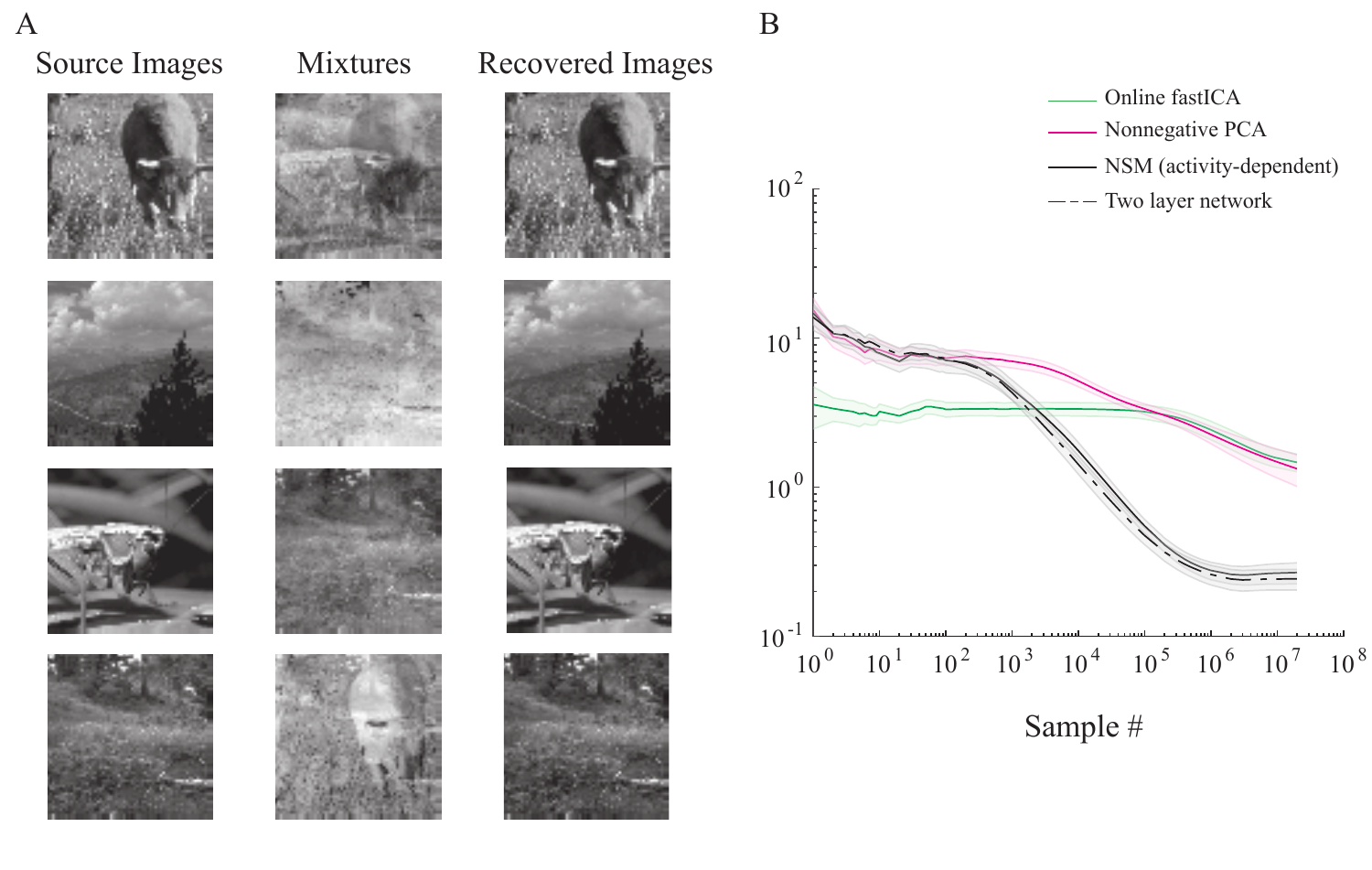}
\caption{Performance of our algorithm when presented with a mixture of natural images: A)  Sample source, mixture, and recovered images, performed by our two-layer algorithm. B) Performance comparison of our online algorithms with Online fastICA and Nonnegative PCA. Shaded bars show standard error over 10 simulations. Learning rate parameters are listed in Appendix \ref{lrate}. \label{Fig6}}
\end{figure}
%



\section{Discussion}

In this paper we presented a new neural algorithm for blind nonnegative source separation. We started by assuming the nonnegative ICA generative model \citep{plumbley2001adaptive,plumbley2002conditions} where inputs are linear mixtures of independent and nonnegative sources. We showed that the sources can be recovered from inputs by two sequential steps, 1) generalized whitening and 2) NSM. In fact, our argument requires sources to be only uncorrelated, and not necessarily independent. Each of the two steps can be performed online with single-layer neural networks with local learning rules \citep{pehlevan2014NMF,pehlevan2015normative}. Stacking these two networks yields a two-layer neural network for blind nonnegative source separation (Fig. \ref{Fig2}). Numerical simulations show that our neural network algorithm performs well.

Because our network is derived from optimization principles, its biologically realistic features can be given meaning. The network is multi-layered, because each layer performs a different optimization. Lateral connections create competition between principal neurons forcing them to differentiate their outputs. Interneurons clamp the activity dimensions of principal neurons \citep{pehlevan2015normative}. Rectifying neural nonlinearity is related to nonnegativity of sources. Synaptic plasticity \citep{malenka2004ltp}, implemented by local Hebbian and anti-Hebbian learning rules, achieves online learning. While Hebbian learning is famously observed in neural circuits \citep{bliss1973longa,bliss1973long}, our network also makes heavy use of anti-Hebbian learning, which can be interpreted as the long-term potentiation of inhibitory postsynaptic potentials. Experiments show that such long-term potentiation can arise from pairing  action potentials in inhibitory neurons with subthreshold
depolarization of postsynaptic pyramidal neurons \citep{komatsu1994age,maffei2006potentiation}. However, plasticity in inhibitory synapses does not have to be Hebbian, i.e. require correlation between pre- and postsynaptic activity  \citep{kullmann2012plasticity}. 

For improved biological realism, the network should respond to a continuous stimulus stream by continuous and simultaneous changes to its outputs and synaptic weights. Presumably, this requires neural time scales to be faster and synaptic time scales to be slower than that of changes in stimuli. To explore this possibility, we simulated some of our datasets with limited number of neural activity updates (not shown) and found that $\sim$10 updates per neuron is sufficient for successful recovery of sources without significant loss in performance. With a neural time scale of 10ms, this should take about 100ms, which is sufficiently fast given that, for example, the temporal autocorrelation time scale of natural image sequences is about 500ms \citep{david2004natural,bull2014communicating}.

It is interesting to compare the two-layered architecture we present to the multilayer neural networks of deep learning approaches \citep{lecun2015deep}. 1) For each data presentation, our network performs recurrent dynamics to produce an output, while the deep networks have feedforward architecture. 2) The first layer of our network has multiple neuron types, principal and interneurons, and only principal neurons project to the next layer. In deep learning, all neurons in a layer project to the next layer. 3) Our network operates with local learning rules, while deep learning uses backpropagation, which is not local. 4) We derived the architecture, the dynamics, and the learning rules of our network from a principled cost function. In deep learning, the architecture and the dynamics of a neural network are designed by hand, only the learning rule is derived from a cost function. 5) Finally, in building a neural algorithm, we started with a generative model of inputs, from which we inferred algorithmic steps to recover latent sources. These algorithmic steps guided us in deciding which single-layer networks to stack. In deep learning, no such generative model is assumed and network architecture design is more of an art. We believe starting from a generative model might lead to a more systematic way of network design. In fact, the question of generative model appropriate for deep networks is already being asked \citep{patel2016probabilistic}.


\subsubsection*{Acknowledgments}

We thank Andrea Giovannucci, Eftychios Pnevmatikakis, Anirvan Sengupta and Sebastian Seung for useful discussions. DC is grateful to the IARPA MICRONS program for support.

\pagebreak

\appendix

\section{Convergence of the gradient descent-ascent dynamics}\label{app_convergence} 

Here we prove that the neural dynamics \eqref{dynamicsdef} converges to the saddle point of the objective function \eqref{whitenOnlineReduced}. Here we assume that ${\bf W}^{HG}$ is full-rank and $l=d$. First, note that the optimum of \eqref{whitenOnlineReduced} is also the fixed point of \eqref{dynamicsdef}. Since the neural dynamics  \eqref{dynamicsdef}  is linear, the fixed point is globally convergent  if and only if the eigenvalues of the matrix
\begin{align}
\left[ \begin{array}{c c} {\bf 0} & -{\bf W}^{HG} \\ {\bf W}^{GH} & -{\bf I} \end{array} \right]
\end{align}
have negative real parts. 

The eigenvalue equation is
\begin{align}
\left[ \begin{array}{c c} {\bf 0} & -{\bf W}^{HG} \\ {\bf W}^{GH} & -{\bf I} \end{array} \right] \left[\begin{array}{c} {\bf x}_1 \\ {\bf x}_2 \end{array} \right] = \lambda \left[\begin{array}{c} {\bf x}_1 \\ {\bf x}_2 \end{array} \right],
\end{align}
which implies
\begin{align}\label{eigeq}
-{\bf W}^{HG}{\bf x}_2 = \lambda {\bf x}_1, \qquad {\bf W}^{GH} {\bf x}_1 = \left(\lambda + 1\right){\bf x}_2 .
\end{align}
Using these relations, we can solve for all the $d+m$ eigenvalues. There are two cases:
\begin{enumerate}
\item $\lambda = -1$. This implies that ${\bf W}^{GH} {\bf x}_1 = 0$ and ${\bf W}^{HG} {\bf x}_2 = {\bf x}_1$.  ${\bf x}_1$ is in the null-space of ${\bf W}^{GH}$. Since ${\bf W}^{GH}$ is $m\times d$ with $m\geq d$, the null-space is $m-d$ dimensional, and one has $m-d$ degenerate $\lambda = -1$ eigenvalues.
\item $\lambda \neq -1$. Substituting for ${\bf x}_2$ in the first equation of \eqref{eigeq}, this implies that ${\bf W}^{HG}{\bf W}^{GH} {\bf x}_1 = -\lambda\left(\lambda+1\right){\bf x}_1$. Hence, ${\bf x}_1$ is an eigenvector of ${\bf W}^{HG}{\bf W}^{GH}$. For each eigenvalue $\lambda_w$ of ${\bf W}^{HG}{\bf W}^{GH}$, there are two corresponding eigenvectors ${\bf \lambda} = -\frac 12 \pm \sqrt{\frac 14 -\lambda_w}$. ${\bf x}_2$ can be solved uniquely from the first equation in \eqref{eigeq}. 
\end{enumerate}
Hence, there are $m-d$ degenerate $\lambda = -1$ eigenvalues and $d$ pairs of conjugate eigenvalues ${\bf \lambda} = -\frac 12 \pm \sqrt{\frac 14 -\lambda_w}$, one pair for each eigenvaleue $\lambda_w$ of ${\bf W}^{HG}{\bf W}^{GH}$. Since $\lbrace \lambda_w\rbrace $ are real and positive (we assume ${\bf W}^{HG}$ is full-rank and by definition ${\bf W}^{HG}={{\bf W}^{GH}}^\top$), real parts of all $\lbrace \lambda\rbrace $ are negative and hence the neural dynamics \eqref{dynamicsdef} is globally convergent.

\pagebreak
\section{Modified objective function and neural network for generalized prewhitening}\label{FullGenWhite}

While deriving our online neural algorithm, we assumed that  the number of output channels is reduced to the number of sources at the prewhitening stage ($l=d$). However, our offline analysis did not need such reduction, one could keep $l\geq d$ for generalized prewhitening. Here we provide an online neural algorithm that allows $l \geq d$.

First, we point out why the prewhitening algorithm given in the main text is not adequate for $l > d$. In Appendix \ref{app_convergence}, we proved that the neural dynamics described by \eqref{dynamicsdef} converges to the saddle point of the objective function \eqref{whitenOnlineReduced}. This proof assumes that  ${\bf W}^{HG}$ is full-rank. However, if $l > d$, this assumption breaks down as the network learns because perfectly prewhitened $\delta{\bf H}$ has rank $d$ (low-rank) and a perfectly prewhitening network would have ${\bf W}^{HG}= \delta{\bf H}\delta{\bf G}^\top$ which would also be low-rank. We simulated this network with $l>d$ and observed that the condition number of  ${\bf W}^{HG} {\bf W}^{GH}$ increased with $t$ and the neural dynamics took longer time to converge. Even though the algorithm was still functioning well for practical purposes, we present a modification that fully resolves the problem.

We propose a modified offline objective function \citep{pehlevan2015normative} and a corresponding neural network. Consider the following:
\begin{align}\label{NIPS3_2}
\min_{\delta{\bf H}} \max_{\delta{\bf G}} {\rm Tr}\left(-\delta {\bf X}^\top\delta{\bf X}\,\delta{\bf H}^\top \delta{\bf H} + \delta{\bf H}^\top \delta{\bf H}\,\delta{\bf G}^\top \delta{\bf G} + \alpha t\, \delta{\bf H}^\top \delta{\bf H} - t \, \delta{\bf G}^\top \delta{\bf G}\right),
\end{align}
where $\delta {\bf X}$ is a $k\times t$ centered mixture of $d$ independent sources, $\delta{\bf H}$ is now an $l\times t$ matrix with $l\geq d$, $\delta{\bf G}$ is an $m\times t$ matrix with $m\geq d$ and  $\alpha$ is a positive parameter. Notice the additional $\alpha$-dependent term compared to \eqref{NIPS3}. If $\alpha$ is less than the lowest eigenvalue of $\frac 1t{\bf \delta X}{\bf \delta X}^\top$, optimal $\delta {\bf H}$ is a linear transform of ${\bf X}$ and satisfies the generalized prewhitening condition \eqref{hhbar}\citep{pehlevan2015normative}. More precisely, 
 \begin{Th}[Modified from \citep{pehlevan2015normative}]Suppose an eigen-decomposition of $\delta{\bf X}^\top\delta{\bf X}$ is $\delta{\bf X}^\top\delta{\bf X} = {\bf V}^X {\bf \Lambda}^X {{\bf V}^X}^\top$, where eigenvalues are sorted in order of magnitude.   If $\alpha$ is less than the lowest eigenvalue of $\frac 1t{\bf \delta X}{\bf \delta X}^\top$, all optimal $\delta{\bf H}$ of \eqref{NIPS3} have an SVD decomposition of the form
 \begin{align}
\delta{\bf H}^* ={\bf U}^H \,\sqrt{t}\, {{\bf \Lambda}^H}' \, {{\bf V}^X}^\top,
\end{align}
where  ${{\bf \Lambda}^H}'$ is $l\times t$ with ones at top $d$ diagonals and zeros at rest.
\end{Th}
Using this cost function, we will derive a neural algorithm which does not suffer from the described convergence issues, even if $l>d$. On the other hand, we now need to choose the parameter $\alpha$, and for that we need to know the spectral properties of ${\delta \bf X}$.

To derive an online algorithm, we repeat the steps taken before:
\begin{align}\label{whitenOnline2_2}
\{\delta{\bf h}_t,\delta{\bf g}_t\} \longleftarrow \mathop {\arg \min }\limits_{\delta{\bf h}_t}  \mathop {\arg \max }\limits_{\delta{\bf g}_t} L(\delta{\bf h}_t,\delta{\bf g}_t),
 \end{align}
 where
\begin{align}\label{whitenOnlineReduced_2}
L =  - 2{\delta{\bf x}^\top_t}\left( {\sum\limits_{t' = 1}^{t - 1} {\delta{\bf x}_{t'}{\delta{\bf h}^\top_{t'}}} } \right)\delta{\bf h}_t -t\left\| \delta{{\bf g}_t} \right\|^2_2 &+ 2{\delta{\bf g}^\top_t}\left( {\sum\limits_{t' = 1}^{t - 1} {\delta{\bf g}_{t'}{\delta{\bf h}^\top_{t'}}} } \right)\delta{\bf h}_t + \alpha t \left\| \delta{{\bf h}_t} \right\|^2_2\nonumber \\
&\qquad\qquad+\left(\left\| \delta{{\bf g}_t} \right\|^2_2 -\left\| \delta{{\bf x}_t} \right\|^2_2\right) \left\| \delta{{\bf h}_t} \right\|^2_2 .
 \end{align}
In the large-$t$ limit, the first four terms dominate over the last term, which we ignore. The remaining objective is strictly concave in $\delta{\bf g}_t$ and strictly convex in $\delta{\bf h}_t$. Note that \eqref{whitenOnlineReduced} was only convex in $\delta{\bf h}_t$ but not strictly convex. The objective has a unique saddle point, even if $ \frac 1t {\sum\limits_{t' = 1}^{t - 1}  \delta{\bf h}_{t'}}{\delta{\bf g}^\top_{t'}} $ is not full-rank:
\begin{align}\label{onlinesaddle_2}
\delta{\bf g}^*_t &= {\bf W}^{GH}_t\delta{\bf h}^*_t, \nonumber \\
\delta{\bf h}^*_{t} &= \left(\alpha {\bf I}+{\bf W}^{HG}_t{\bf W}^{GH}_t\right)^{-1}{\bf W}^{HX}_t\delta{\bf x}_t, 
\end{align}
 where ${\bf W}$ matrices are defined as before and ${\bf I}$ is the identity matrix.

We solve \eqref{whitenOnline2_2} with gradient descent-ascent
\begin{align}\label{dynamicsdef_2}
\frac{d\delta{\bf h}_t}{d\gamma} &= -\frac{1}{2t}\nabla_{{\delta{\bf h}_t}}{L} =- \alpha \delta{\bf h}_t+{\bf W}^{HX}_t\delta{\bf x}_t - {\bf W}^{HG}_t{\bf g}_t, \nonumber \\
\frac{d\delta{\bf g}_t}{d\gamma} &= \frac{1}{2t}\nabla_{{\delta{\bf g}_t}}{L} = - \delta{\bf g}_t + {\bf W}^{GH}_t\delta{\bf h}_t.
\end{align}
where $\gamma$ is time measured within a single time step of $t$.
The dynamics \eqref{dynamicsdef} can be proved to converge to the saddle point \eqref{onlinesaddle_2} modifying the proof in Appendix \ref{app_convergence}\footnote{The fixed point is globally convergent  if and only if the eigenvalues of the matrix
\begin{align}
\left[ \begin{array}{c c} -\alpha{\bf I} & -{\bf W}^{HG} \\ {\bf W}^{GH} & -{\bf I} \end{array} \right]
\end{align}
have negative real parts. One can show that $l-d$ eigenvalues are $-\alpha$, $m-d$ eigenvalues are $-1$, and for each positive eigenvalue, $\lambda_w$ of ${\bf W}^{HG}{\bf W}^{GH}$ one gets a pair $-\frac{1+\alpha}{2}\pm \sqrt{\frac{\left(1+\alpha\right)^2}{4}-\alpha-\lambda_w}$. All eigenvalues have negative real parts.
}. Synaptic weight updates are the same as before \eqref{updatedef}. Finally, this network can be modified to also compute $\bar {\bf H}$ following the steps before.

\pagebreak
\section{Mixing matrices for numerical simulations}\label{mixing}

For the random source dataset, the $d=3$ mixing matrix was:
\begin{align}
A = \left(\begin{array}{rrr} 0.031518 & 0.38793 & 0.061132 \\ -0.78502 & 0.16561 & 0.12458 \\ 0.34782 & 0.27295 & 0.67793 \end{array}\right),
\end{align}
We do not list the mixing matrices for $d=\lbrace 5,7, 10 \rbrace$ cases for space-saving purposes, however they are available from authors upon request.

For the natural scene dataset, the mixing matrix was
\begin{align}
A = \left(\begin{array}{rrrr} 0.33931	& 0.3282	&0.41516	&-0.1638\\
-0.077079	 & -0.29768	&0.076562	&-0.28153 \\  
-0.14119	&-0.41709 	&0.14842	&0.57009 \\  
-0.40483	&0.21922	&0.082336	&0.18027 \end{array}\right).
\end{align}

\pagebreak
\section{Learning rate parameters for numerical simulations}\label{lrate}

For Figs. \ref{Fig3}, \ref{Fig4}, \ref{Fig5} and \ref{Fig6} the following parameters were found to be best performing as a result of our grid search:

\begin{tabular}{c || c | c | c | c |  }	
   & fastICA & NPCA & NSM (activity) & NSM (time)  \\\hline\hline
  $d=3$ & (10, 0.01) & (10, 0.1) & (10, 0.8) & (10, 0.1)   \\ \hline
  $d=5$ &  (100, 0.01) &  (10, 0.01) & (10, 0.9) & (10, 0.01) \\ \hline
  $d=7$ & (100, 0.01) & (100, 0.01) & (10, 0.9) & (10, 0.1)  \\ \hline
  $d=10$ & (100, 0.01) & (1000, 0.01) & (10, 0.9) & (10, 0.01)\\ \hline
  {\rm Images} & ($10^4$, 0.01) & (1000, 0.01) & (100, 0.9) & NA  \\
  \hline  
  
\end{tabular}

\vskip 10pt
\begin{tabular}{c || c | c | c | c |  }	
   &Two-layer & Offline & Infomax ICA & Linsker's Algorithm \\\hline\hline
  $d=3$ &  (100, 1, 10, 0.8) & ($10^6$, 0.001) & (1000, 0.2) & (1000, 0.2) \\ \hline
  $d=5$ &   (100, 1, 10, 0.9) & ($10^6$, 0.001)& (1000, 0.2) & (1000, 0.2)\\ \hline
  $d=7$ &  (100, 1, 10, 0.9) & ($10^6$, 0.01)& (1000, 0.2) & (1000, 0.2) \\ \hline
  $d=10$ &  (100, 1, 10, 0.9) & ($10^6$, 0.01)& (1000, 0.2) & (1000, 0.2)\\ \hline
  {\rm Images}  & (100, 1, 100, 0.9) & NA  & NA & NA\\
  \hline  
  
\end{tabular}

\pagebreak

\end{document}